\pdfoutput=1
\documentclass[runningheads]{llncs}

\usepackage[titletoc]{appendix}

\usepackage {amssymb}
\usepackage {amsmath}
\usepackage [usenames] {xcolor}
\usepackage {enumerate}
\usepackage{graphicx}  
\usepackage{hyperref}
\usepackage{algorithm}
\usepackage{algorithmic}
\usepackage{subcaption}

\title {How to Fit a Tree in a Box}

\author {Hugo A. Akitaya\inst{1} \and Maarten L\"offler\inst{2} \and Irene Parada\inst{3}}
\institute {
Tufts University, Massachusetts, USA\\
\email{hugo.alves\_akitaya@tufts.edu}
\and
Utrecht University, Utrecht, The Netherlands\\
\email{m.loffler@uu.nl}
\and 
Graz University of Technology, Graz, Austria\\
\email{iparada@ist.tugraz.at}
}

\begin{document}
\maketitle
\begin{abstract}
  We study compact straight-line embeddings of trees. We show that perfect binary trees can be embedded optimally: a tree with $n$ nodes can be drawn on a $\sqrt n$ by $\sqrt n$ grid. We also show that testing whether a given binary tree has an upward embedding with a given combinatorial embedding in a given grid is NP-hard.
\end{abstract}




\section {Introduction}

Let $T = (V,E)$ be combinatorial tree; that is, a connected graph without cycles.
A {\em straight-line embedding} of $T$ onto a grid is an injective map $f : V \to \mathbb{Z}^2$.
An embedding is {\em planar} if for every pair of edges $(v_1,v_2), (w_1,w_2) \in E$ the line segments $f(v_1)f(v_2)$ and $f(w_1)f(w_2)$ do not intersect except at common endpoints.
The {\em size} or {\em dimensions} of an embedding (or, with slight abuse of terminology, the size of the grid) is the width and height of the portion of $\mathbb{Z}^2$ used by $f$; that is,
$$
  \dim_f (T) =
  \left( \max_{v \in V} x_{f(v)} - \min_{v \in V} x_{f(v)} + 1, 
         \max_{v \in V} y_{f(v)} - \min_{v \in V} y_{f(v)} + 1
  \right).
$$
We are interested in finding embeddings with as small a size as possible.

A {\em rooted} tree is a tree $T$ with a special vertex $r \in V$ marked as root.
Because a tree has no cycles, a rooted tree has an induced partial order on its vertices: for two vertices $v, w \in V$, we say $v \prec w$ if and only if $v$ lies on the path from $r$ to $w$.
An embedding is {\em upward} if, for all $v, w \in V$ with $v \prec w$, we have $y_{f(v)} > y_{f(w)}$.
An embedding is {\em weakly upward} if, for all $v, w \in V$ with $v \prec w$, we have $y_{f(v)} \ge y_{f(w)}$.
\bigskip

\noindent\textbf{Related Work.}
Drawing graphs with small area has a long and rich history~\cite{survey_layouts}. 
By now, we are starting to have some understanding of when graphs admit drawings with {\em linear area} (a graph with $n$ nodes can be embedded on a $w \times h$ grid with $wh \in O(n)$), and when superlinear area is required.
Chan~\cite{trees_2018} shows that every tree admits a drawing with $n2^{O(\sqrt{\log\log n \log\log\log n})}$ area, improving the long-standing $O(n \log n)$ bound one obtains by a simple divide-and-conquer layout algorithm.

However, not much is known about the exact minimum area requirements for graphs that do admit linear-area drawings.
It is clear that not every tree admits a perfect drawing on a grid with exactly $n$ points: for instance, when the graph is a star, some grid points are ``blocked'' and cannot be used.
The star graph {\em can} be drawn on a linear-area grid: Euler already showed that the fraction of points visible from the center of a square grid tends to $\frac 6 {\pi^2}$ more than 300 years ago~\cite{euler}.
For graphs of bounded degree, there is hope that we can do better. Clearly, every path admits a perfect drawing.
Garg and Rusu~\cite{gen_linear_area,bin_linear_area} show that trees of degree $d = O(n^{\delta})$ with $\delta < 1/2$, and in particular of degree 3, have linear-area drawings onto a square grid, and even onto grids of different aspect ratio; their main concern is studying the relation between the aspect ratio and the area, but they do not give concrete bounds on the constant factor.

We conjecture that every degree $3$ tree admits a {\em perfect} drawing onto a square grid, and we prove here that this is the case for perfect binary trees.


\medskip

When drawing rooted trees, a natural restriction is to require drawings to be {\em upward}.
In this case, clearly, perfect drawings are impossible unless the tree is a path, but we may still investigate almost-perfect drawings that leave only few grid points unused.
Chan~\cite{trees_2018} shows that for strictly upward drawings, we cannot do better than $\Theta (n \log n)$ area. He does give an improved bound for {\em weakly} upward drawings.

Biedl and Mondal~\cite{upward_trees} proved NP-hardness for strictly upward unordered straight-line high-degree trees.
Later, Biedl~\cite{opt_tree_2017} gave an algorithm to find for every ternary tree $T$ a strictly upward order-preserving straight-line grid drawing of optimum width. 

\bigskip

\noindent\textbf{Contribution.}
  We have the following results.
  \begin {itemize}
    \item It is NP-hard to test whether binary trees with fixed combinatorial embedding  admit upward drawings on a given grid.
    \item Perfect binary trees with $n$ vertices admit drawings on a $\sqrt n \times \sqrt n$ grid.
  \end {itemize}

\section {Optimal embeddings of perfect binary trees}

We consider the following setting. Given a $\sqrt n \times \sqrt n$ grid and a tree with $n$ vertices, can we draw it with straight non-crossing edges?
Clearly this is not always possible, for instance if the tree is a star.

\begin {conjecture}
  If the tree has max degree 3, it is always possible.
\end {conjecture}

In particular, if $n=2^{k+1}$, a perfect binary tree of odd height $k$ with additional parent of the root (to make the number of vertices exactly $n$) can be drawn on the $\sqrt n \times \sqrt n$ grid.
We use a recursive strategy to show it.
Similar approaches recursively embedding trees have been previously used to show 
asymptotic bounds (but disregarding smaller order terms);
in particular to prove that
perfect binary trees and Fibonacci trees can be upward drawn in linear area~\cite{up_linear_1992}
and 
to bound the area of complete ternary and 7-ary trees on the 8-grid~\cite{k-grids}.

\begin {theorem} \label {thm:bintree}
  The perfect binary tree on $n=2^{k+1}-1$ vertices with $k$ odd can be embedded in the $\sqrt n \times \sqrt n$ square  grid.
\end {theorem}

\begin{proof}
  We will recursively argue that perfect binary trees can be embedded in square grids in two ways. Let $T_k$ be the perfect binary tree on $n=2^{k+1}-1$ vertices.
  We will recursively define two straight-line crossing-free drawings, $F_k$ and $G_k$, of $T_k$. 
  The vertices in these drawings are placed in the grid points $\{(x,y) \in \mathbb{Z}^2:\ 1\leq x \leq 2^{(k+1)/2}, \ 1\leq y \leq 2^{(k+1)/2} \}$.
  
  We first list the required properties of $F_k$ and $G_k$, also illustrated in Fig.~\ref{fig:FG}:
\begin{enumerate}[(i)]
  \item both $F_k$ and $G_k$ map the root of $T_k$ to the point $(2^{(k-1)/2} + 1, 2^{(k-1)/2})$;
  \item both $F_k$ and $G_k$ do not place any edges in the vertical strip between $x=2^{(k-1)/2}$ and $x=2^{(k-1)/2} + 1$, except for the edges incident to the root of $T_k$; 
  \item $F_k$ leaves the point $(2^{(k-1)/2}, 1)$ unused; and
  \item $G_k$ leaves the point $(1, 1)$ unused.
\end{enumerate}
  
\begin{figure*}[tb]
	\centering 
	\begin{subfigure}[b]{\textwidth}
		\centering
		\includegraphics[width=\linewidth]{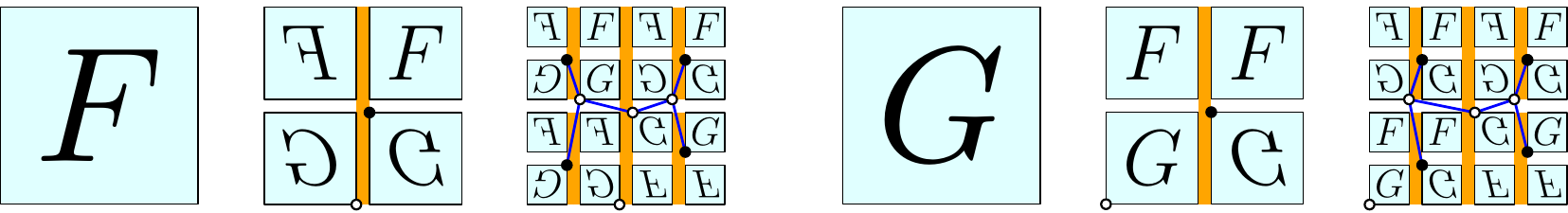}
		\caption{Tiles $F$ and $G$, and their recursive definition.}
		\label{fig:FG}
	\end{subfigure}
    \bigskip
	
	\begin{subfigure}[b]{\textwidth}
		\centering
		\includegraphics[width=\linewidth]{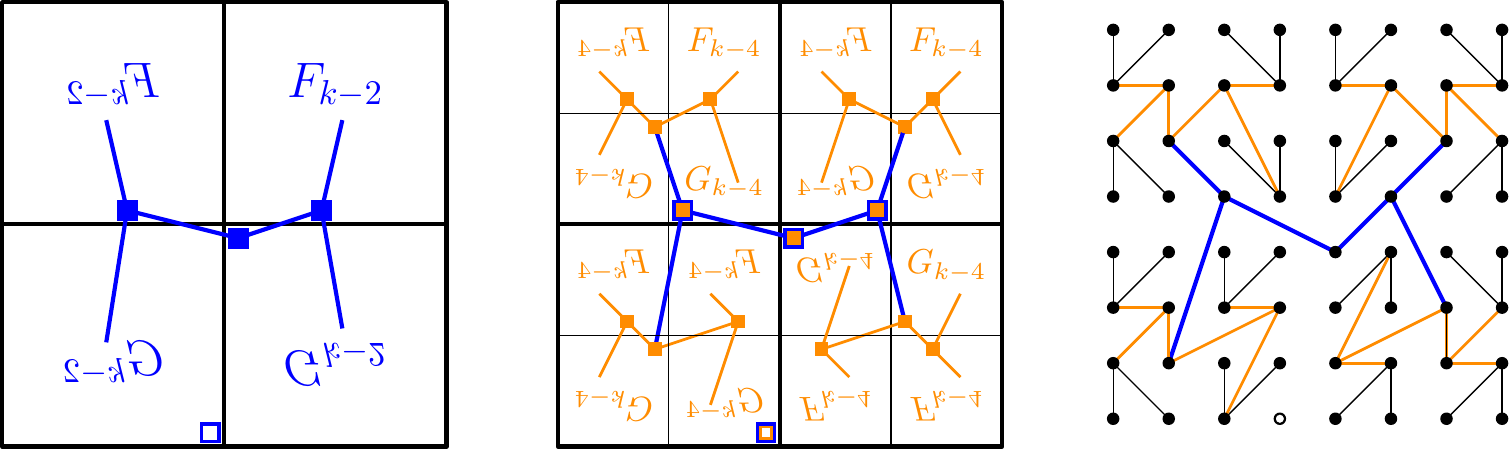}
		\caption{Right and center: recursive definition of $F_k$. Left: $F_5$.}
		\label{fig:Fk}
	\end{subfigure}
	\caption{Recursive embedding of perfect binary trees.}\label{fig:complete_trees}
\end{figure*}
  
  Observe that $F_1 = G_1$ is trivial to draw: both are drawings of a path of length $2$, drawn by connecting the point $(1,2)$ to the point $(2,1)$ to the point $(2,2)$.
  
  What remains is to argue that we can recursively draw $F_k$ and $G_k$ using drawings of $F_{k-2}$ and $G_{k-2}$. 
  The argument is illustrated in Fig.~\ref{fig:complete_trees}; the full proof can be found in Appendix~\ref{app:FG}.  \qed
\end{proof}

\section {Upward embedding of trees in a given grid is NP-Hard}
\label{sec:hardness}

Recall that an embedding of a rooted tree is \emph{upward} if the $y$-coordinate of a node is strictly greater than the $y$-coordinate of its children.
A \emph{combinatorial embedding} is given by a circular order of incident edges around each vertex.
In this section we show that deciding if a rooted binary tree with a fixed combinatorial embedding can be drawn upward and without crossings in a given square grid is NP-complete. 

\begin {theorem} \label {thm:upward}
  Deciding whether an upward planar straight-line drawing of a fixed combinatorial embedding of a rooted binary tree on a grid of given size $(w \times h)$ exists is NP-complete.
\end {theorem}
\begin{proof}
The problem is in NP since a geometric drawing of a tree with $k$ vertices in the grid can be expressed in $O(k)$ size by assigning vertices to grid points. Checking whether the drawing is an embedding can trivially be done in $O(k^2)$ time by checking pairwise edge crossings. 
Checking whether the drawing preserves the given rotation system takes $O(k^2)$ time and 
checking whether it is upward can be done in $O(k)$ time.

We prove NP-hardness by a reduction from 3SAT which is an NP-complete problem~\cite{karp_1972}.
An instance of 3SAT is given by a set $\{x_1,\ldots,x_n\}$ of $n$ variables and a set $\{c_1,\ldots,c_m\}$ of $m$ clauses.
Each variable can assume one of two values in $\{\texttt{true}, \texttt{false}\}$.
Each clause is defined by 3 literals, i.e., positive or negative copies of a variable.
A clause is \emph{satisfied} if at least one of its literals is \texttt{true}.
The problem 3SAT asks for an assignment from the variables to $\{\texttt{true}, \texttt{false}\}$ that satisfies all clauses.
We give an arbitrary order for the variables and say that $x_i$ appears \emph{before} $x_j$ if $i<j$.
The first (resp., second, resp., third) literal of a clause is the literal (among the 3 literal that define the clause) of the variable that appears first (resp., second, resp., third) in the order assigned to variables.
Given an instance of 3SAT we build a rooted tree with $O(m^2+mn)$ vertices and set $w=4m+4$ and $h=\lceil \lg(4m+4)\rceil+5n+4m+1$.

\begin{figure}[h]
	\centering
	\includegraphics[width=0.75\linewidth]{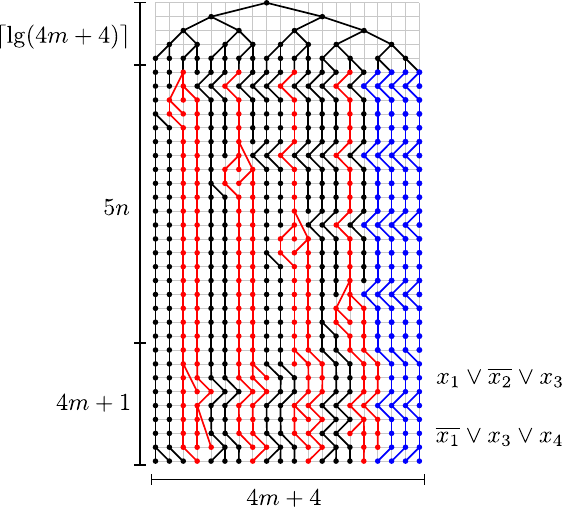}
	\caption{Reduction from 3-SAT.}
	\label{fig:reduction}
\end{figure}

\medskip

\noindent\textbf{Overview.}
Refer to Fig.~\ref{fig:reduction}.
This paragraph gives a brief informal overview of the reduction. 
The following paragraphs will give a full proof.
The reduction is divided into 3 parts. 
The top part (spanning the top $\lceil \lg(4m+4)\rceil$ rows in Fig.~\ref{fig:reduction}) is a perfect binary tree with $2^{\lceil \lg(4m+4)\rceil-1}$ leaves.
The middle part (spanning the next $5n$ rows in Fig.~\ref{sec:hardness}) is where the variables are assigned a boolean value.
The bottom part (spanning the last $4m+1$ rows in Fig.~\ref{sec:hardness}) enforces that every clause of the original instance of 3SAT is satisfied.
Each variable is represented by a red subtree with two long paths that have to span all but one row below the least common ancestor.
The left (resp., right) path represents a positive (resp., negative) literal of the variable.
The construction forces one of the paths to be drawn one unit above the other and that encodes the boolean assignment. 
If the left path does not span the last row, then the variable is set to \texttt{true}. 
The variable is set to \texttt{false} otherwise.
In Fig.~\ref{fig:reduction}, $x_2$ and $x_3$ (resp., $x_1$ and $x_4$) are set to \texttt{true} (resp., \texttt{false}).
The blue subtrees encode the clauses by allowing the rest of the construction to occupy specific extra grid positions.
The incidence of a variable in a clause is encoded by an extra leaf child in one of the paths that represent the incident literal corresponding to the variable.
If none of the incident literals of a clause are set to \texttt{true}, the drawing would require the use of an extra row or column.
Otherwise, the extra leaves can be accommodated exactly by the space provided by the blue subtrees.

\medskip

\noindent\textbf{Construction.}
There are exactly $4m+4$ subtrees attached to the perfect binary tree on the top of the construction.
The fixed combinatorial embedding prescribes a left-to-right order of such subtrees.
For each variable $x_i$, do the following.
Set the $4(i-1)+1$-th subtree to be a path $p$ of length $5n+4m$; attach another path of length $5n+4m-5(i-1)-4$ to the right of the $5(i-1)+4$-th vertex of $p$.
Attach a right child to the second to last vertex of $p$.
Set the $4(i-1)+2$-th subtree to be a path of length $5(i-1)+1$.
Set the $4(i-1)+3$-th subtree to be a path of length $5(i-1)$.
At the end of the path, attach two paths $p_t$ and $p_f$ of length $5(n-i+1)+4m-2$ each  as left and right subtrees respectively. 
Attach a right (resp., left) child to the first vertex of $p_t$ (resp., $p_f$).
We now describe the position of the vertices that encode the incidence of a variable in a clause. 
We call such vertices \emph{literal leaves}.
If $x_i$ (resp., $\overline{x_i}$) is the first or second literal of $c_j$, then add a right child $\ell_{i,j}$ to the $5(n-i+1)+4(m-1)$-th vertex of $p_t$ (resp., $p_f$).
If $x_i$ (resp., $\overline{x_i}$) is the third literal of $c_j$, then add a left child $\ell_{i,j}$ to the $5(n-i+1)+4(m-1)$-th vertex of $p_t$ (resp., $p_f$).
The $4(i-1)+3$-th subtree is shown in red in Fig.~\ref{fig:reduction}.
Set the $4(i-1)+4$-th subtree to be a path of length $5(i-1)$.
Finally, we describe the four last subtrees (shown in blue in Fig.~\ref{fig:reduction}).
Set the $4m+1$-th, $4m+2$-th, and $4m+3$-th subtrees to be paths of length $5n+4m$ each.
For every clause $c_j$, attach a right leaf child to the $5n+4j$-th vertex of the $4m+3$-th subtree.
Set the last subtree to be a path of length $5n$.
For each variable $x_i$, attach a right leaf child to the $5i-4$-th vertex of the path.
This finalizes the construction.
\medskip

\noindent\textbf{Correctness.}
We argue that the construction is correct,
by showing that every satisfiable 3SAT instance can indeed be embedded in a $w \times h$ grid, and that every drawing that fits in a $w \times h$ grid must correspond to a satisfiable 3SAT instance. The full details can be found in Appendix~\ref {app:up}.
\qed
\end{proof}

\section {Conclusions}

We studied tree drawings in small areas.

For arbitrary drawings, we gave a construction for embedding a perfect binary tree on a square grid. The main remaining open question here is whether every low-degree tree with $wh$ vertices (or fewer) can be embedded on an $w \times h$ grid. We conjecture that this is the case when $w=h$.
Another intriguing question is whether, for general trees, testing if they can be embedded on a given grid is computationally tractable.

For upward drawings, we showed that even for bounded-degree trees, testing whether a given tree can be embedded in a $w \times h$ rectangle is already NP-hard, if the combinatorial embedding of the tree is fixed. It would be interesting to know whether the same is true when one can freely choose the combinatorial embedding.
Another question is whether the problem is also NP-hard for {\em weakly} upward drawings, where adjacent vertices may be embedded using the same $y$-coordinate.

\section* {Acknowledgements}
This research was initiated at the 33rd Bellairs Winter Workshop on Computational Geometry in 2018. We would like to thank all participants of the workshop for fruitful discussions on the topic.
H.A.A. was supported by NSF awards CCF-1422311 and CCF-1423615, and the Science Without Borders scholarship program.
M.L. was partially supported by the Netherlands Organisation for Scientific Research (NWO) through grant number 614.001.504.
I.P. was supported by the Austrian Science Fund (FWF) grant W1230.


\appendix

\section {Full Proof of Theorem~\ref {thm:bintree}} \label {app:FG}

  We argue that the construction in the proof of Theorem~\ref {thm:bintree} is correct,
  by showing that we can recursively draw $F_k$ and $G_k$ using drawings of $F_{k-2}$ and $G_{k-2}$. The argument below is illustrated in Fig.~\ref{fig:complete_trees}.

  To draw $F_k$ we place two (possibly mirrored) copies of $F_{k-2}$ and two (possibly mirrored) copies of $G_{k-2}$ 
  in the four quadrants of the $2^{(k+1)/2}$ by $2^{(k+1)/2}$ grid: 
  one identical copy of $F_{k-2}$ in the top right quadrant,  
  one horizontally mirrored copy of $F_{k-2}$ in the top left quadrant, 
  one vertically mirrored copy of $G_{k-2}$ in the bottom right quadrant, 
  and one horizontally mirrored copy of $G_{k-2}$ in the bottom right quadrant. 
  Because, except for edges incident to the respective roots, each copy leaves a vertical strip empty, 
  we can connect the roots of the top right and bottom right subtrees to a new node at $(2^{(k-1)/2} + 2^{(k-3)/2}, 2^{(k-1)/2} + 1)$ (which was left empty by definition of $F$).
  Note that, since the new edges are also incident to the roots of the subtrees, 
  they don't cross with the edges incident to the respective roots that cross the strip.
  Moreover, the new top edge lies completely above the edge crossing the strip in the bottom right quadrant
  and the new bottom edge lies completely below the edge crossing the strip in the top right quadrant.
  Thus, the two new edges are drawn without crossings.
  Similarly, we connect the roots of the top left and bottom left subtrees to a new node at $(2^{(k-3)/2} + 1, 2^{(k-1)/2} + 1)$. 
  Finally, we connect both of these new nodes to the new root of $T_k$, 
  drawn at $(2^{(k-1)/2}+1, 2^{(k-1)/2})$ (which was left empty by definition of $G$), as required. 
  The point $(2^{(k-1)/2}, 1)$ remains unused, 
  and the central vertical strip is empty except for one edge connecting the root of $T_k$ to the left subtree.
  
  To draw $G_k$ we also place two copies of $F_{k-2}$ and two copies of $G_{k-2}$ 
  in the four quadrants of the $2^{(k+1)/2}$ by $2^{(k+1)/2}$ grid:
  one identical copy of $F_{k-2}$ in the top right quadrant,
  another identical copy of $F_{k-2}$ in the top left quadrant, 
  one identical copy of $G_{k-2}$ in the bottom left quadrant,
  and one vertically mirrored copy of $G_{k-2}$ in the bottom right quadrant.
  As in the case of $F$, we can connect the roots of the top right and bottom right subtrees to a new node at $(2^{(k-1)/2} + 2^{(k-3)/2}, 2^{(k-1)/2} + 1)$, 
  and the roots of the top left and bottom left subtrees to a new node at $(2^{(k-3)/2}, 2^{(k-1)/2} + 1)$. 
  We connect the two new nodes to the new root of $G_k$ at $(2^{(k-1)/2} + 1, 2^{(k-1)/2})$. 
  The point $(1, 1)$ remains unused, 
  and the central vertical strip is again empty except for the edge connecting the root of $G_k$ to the left subtree.

\section {Full Proof of Theorem~\ref {thm:upward}} \label {app:up}

We argue that the hardness construction in the proof of Theorem~\ref {thm:upward} is correct,
by showing that every satisfiable 3SAT instance can indeed be embedded in a $w \times h$ grid, and that every drawing that fits in a $w \times h$ grid must correspond to a satisfiable 3SAT instance.
\medskip

\noindent\textbf{Correctness $(\Rightarrow)$.}
First assume that the 3SAT instance has a positive solution.
Then we can upward embed the constructed tree in a $(w\times h)$ grid as follows.
Use the first $\lceil \lg(4m+4)\rceil$ rows to embed the top part of the tree (perfect binary tree).
We describe the embedding by assigning grid points to vertices of the subtrees from left to right.
Let the bottom left grid point be $(0,0)$.
Assign the root of the $s$-th subtree to $(s,5n+4m)$.
We describe the embedding from left to right with $s$ starting with $1$ until $4m+4$.

\paragraph{General rule.}
Apart from the roots of the paths $p_t$ and $p_f$ of the red subtrees, and the literal leaves, recursively assign the left child to the leftmost free grid point in the row immediately below its parent, then assign the right child to the leftmost free grid point in the same row.

\paragraph{Encoding truth assignment.} 
For the $i$-th red subtree (i.e., $s=4i-1$), let $v_i$ be the least common ancestor of $p_t$ and $p_f$.
Denote by $(v_i.x,v_i.y)$ the coordinate of the grid point assigned to $v_i$.
If $x_i$ is assigned \texttt{true} (resp., \texttt{false}) in the solution of the 3SAT instance, respectively assign the roots of $p_t$ and $p_f$ to $(v_i.x,v_i.y-1)$ and $(v_i.x+1,v_i.y-2)$ (resp., $(v_i.x-1,v_i.y-2)$ and $(v_i.x,v_i.y-1)$).

\begin{figure}[b]
	\centering
\includegraphics[width=0.5\linewidth]{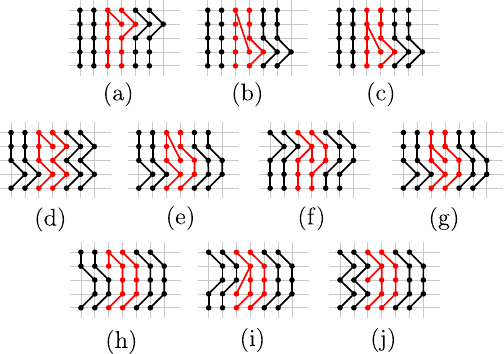}
	\caption{Possible embeddings for the literal leaves of a clause.}
	\label{fig:lit}
\end{figure}

\paragraph{Position of literal leaves $\ell_{i,j}$.}
Let $a$ be the parent of $\ell_{i,j}$.
If $x_i$ or $\overline{x_i}$ is the first literal of a clause $c_j$, we place $\ell_{i,j}$ as follows.
If it is the last (and therefore the only) literal of $c_j$ with \texttt{true} value, assign it to $(a.x+1,a.y-1)$ (see Fig.~\ref{fig:lit}(a)).
Else, if its value is \texttt{true}, assign it to $(a.x+1,a.y-3)$ (see Fig.~\ref{fig:lit}(b)).
Else, assign it to $(a.x+1,a.y-2)$ (see Fig.~\ref{fig:lit}(c)).
If $x_i$ or $\overline{x_i}$ is the second literal of a clause $c_j$, we proceed as follows.
If it is the last \texttt{true} literal of $c_j$, assign $\ell_{i,j}$ to $(a.x+1,a.y-1)$ (see Fig.~\ref{fig:lit}(d)).
Else, if its value is \texttt{true}, assign it to $(a.x+1,a.y-2)$ (see Fig.~\ref{fig:lit} (e)).
If the first literal of $c_j$ was embedded as in Fig.~\ref{fig:lit}(a), then assign $\ell_{i,j}$ to $(a.x,a.y-1)$ (see Fig.~\ref{fig:lit}(f)).
Else, assign $\ell_{i,j}$ to $(a.x+1,a.y-1)$ (see Fig.~\ref{fig:lit}(g)).
If $x_i$ or $\overline{x_i}$ is the third literal of a clause $c_j$, we proceed as follows.
If its value is \texttt{true}, assign $\ell_{i,j}$ to $(a.x,a.y-1)$ (see Fig.~\ref{fig:lit}(h)).
Else, if the second literal of $c_j$ was embedded as in Fig.~\ref{fig:lit}(f) (which happens when the first and second literals are \texttt{true} and \texttt{false} respectively), then assign $\ell_{i,j}$ to $(a.x-1,a.y-2)$ (see Fig.~\ref{fig:lit}(i)).
Else, assign $\ell_{i,j}$ to $(a.x-1,a.y-1)$ (see Fig.~\ref{fig:lit}(j)).
By construction, the blue subtrees will occupy one less grid points in 3 rows of $y$-coordinate $1+4(m-j)$, $2+4(m-j)$ and $3+4(m-j)$ for each $c_j$.
Since the SAT instance has a positive solution, these will be exactly filled by the literal leaves of the corresponding variable gadgets.

\medskip

\noindent\textbf{Correctness $(\Leftarrow)$.}
Now, assume that the constructed tree can be upward embedded in the $(w\times h)$ grid.
We show that the 3SAT instance has a positive solution.

\paragraph{Top part.}
Notice that the number of vertices in the subtrees of the descendants of the top perfect binary tree is $(5n+4m+1)(4m+4)$.
The uppermost $y$-coordinate that each root of a subtree can occupy is $5n+4m$ because each has exactly $\lceil \lg(4m+4)\rceil$ ancestors.
Then, the solution must place all $4m+4$ roots of such subtrees on the row with $y$-coordinate $5n+4m$ and 
$$
(\star) \text{ every point of the grid below this $y$-coordinate must be occupied.}$$

\paragraph{Frame.} The first subtree can only be embedded as shown in Fig.~\ref{sec:hardness}, i.e., each node occupying the uppermost leftmost grid point possible, or else a grid point would not be used contradicting ($\star$).
The same is valid for the blue subtrees as well switching leftmost for rightmost, by ($\star$).
The second tree also has only one possible embedding by ($\star$): if a vertex is mapped to a point that is to the right of its leftmost possible position, then a grid point would be unused; else if it is mapped to a point lower than its uppermost possible position, then the edge between the vertex and its parent would go through a grid point, making it unusable by other vertices. 

\paragraph{Variable assignment.} Let $T_1$ be the first red subtree.
We show that the top part of $T_1$ has only two possible embeddings.
First, note that due to a ``bump" created by a leaf of the last subtree, every grid point $(b, 4m+5n-1)$ for $b\ge 3$ is occupied by a vertex of a tree to the right of $T_1$, or else $(\star)$ would be violated.
Let $v_l$ and $v_r$ be the roots of $p_l$ and $p_r$ respectively (i.e., children of the root of $T_1$).
Then, only one among $v_l$ and $v_r$ can be at $(2, 4m+5n-1)$ and those are the only vertices that can be assigned to that grid point.
If we choose $v_l$ (resp., $v_r$) to be at $(2, 4m+5n-1)$, then $v_r$ (resp., $v_l$) must be at $(1, 4m+5n-2)$ (resp., $(3, 4m+5n-2)$).
The remainder of the embedding of $T_1$ is fixed by $(\star)$ until row $4m$.
Using similar arguments, we can show that, if $(\star)$ is satisfied, once we fix the embedding of the top part of $T_1$, the embeddings of all subtrees are fixed from their roots to row $4m+5n-5$.
We can then apply induction to show that there are only two possible embeddings for the top part of $i$-th red subtree $T_i$ and every non-red subtree has a fixed embedding until row $4m$.

\paragraph{Clause satisfaction.} Once the middle part of the construction is fixed, row $0$ is also fixed because the $4m+3$ leftmost vertices at row $4m$ have $4m+4$ descendant leaves at distance $4m$.
That implies that every vertex on a path of length $4m$ or $4m-1$ in the bottom part of the construction has their $y$-coordinate fixed.
Because of the fixed position of the leaves of the blue trees, the literal leaves of a clause $c_j$ can only occupy the $y$-coordinates $4(m-j)+1$, $4(m-j)+2$, and $4(m-j)+3$. 
Note that a literal leaf can only occupy a $y$-coordinate $4(m-j)+3$ if the corresponding red path was embedded above the other red path of the same variable. 
Since every clause must have a literal leaf at $y$-coordinate $4(m-j)+3$, assigning \texttt{true} (resp., \texttt{false}) to $x_i$ if the corresponding $p_l$ (resp., $p_r$) is embedded above $p_r$ (resp., $p_l$) will result in a satisfying assignment.

\end {document}